\documentclass[a4paper,11pt]{article}
\usepackage{jheppub} 
\usepackage{lineno}
\usepackage{amsmath}
\usepackage{amsthm}
\usepackage{dsfont}
\usepackage{bm}

\newcommand{\Real}{\mathbb{R}}

\newcommand{\Comp}{\mathbb{C}}
\newcommand{\seq}{\subseteq}
\newcommand{\pspace}{\Real^3 \setminus \{0\}}

\newcommand{\Lightcone}{\mathcal{L}_+}

\newcommand{\Poincare}{Poincar\'{e} }
\newcommand{\SO}{\mathrm{SO}}
\newcommand{\so}{\mathfrak{so}}

\newcommand{\ISO}{\mathrm{ISO}}
\newcommand{\kvec}{\boldsymbol{k}}

\newcommand{\mbf}[1]{\boldsymbol{#1}}
\newcommand{\hatbm}[1]{\hat{\bm{#1}}}
\newcommand{\vhat}{\hatbm{v}}
\newcommand{\khat}{\hatbm{k}}
\newcommand{\etheta}{\uv{e}_\theta}
\newcommand{\ephi}{\uv{e}_\phi}
\newcommand{\eh}{\uv{e}_h}
\newcommand{\epm}{\uv{e}_\pm}
\newcommand{\ex}{\uv{e}_x}
\newcommand{\ey}{\uv{e}_y}

\newcommand{\uv}[1]{\mbf{#1}} 

\newcommand{\Phat}{\hatbm{P}} 
\newcommand{\proj}{\mathcal{P}}

\newtheorem{theorem}{Theorem}

\newtheorem{lemma}[theorem]{Lemma}

\arxivnumber{2505.12765} 

\title{\boldmath Spin-weighted spherical harmonics as massless angular momentum eigenstates and their role in obstructing spin-orbital decompositions}

\author[1]{E. Palmerduca\note{Corresponding author.}}
\author{and H. Qin}
\affiliation{Department of Astrophysical Sciences, Princeton University, Princeton, NJ 08544, USA}
\affiliation{Princeton Plasma Physics Laboratory, Princeton, NJ 08543, USA}

\emailAdd{ep11@princeton.edu}

\abstract{We show that for massless helicity $h$ particles, the angular momentum eigenstates are given in an appropriate coordinate system by the spin-weighted spherical harmonics ${_{-h}Y_{jm}}$ of spin-weight $-h$. In particular, these are simultaneous eigenstates of the Hamiltonian, helicity, $J^2$, and $J_z$. The appearance of the spin-weighted spherical harmonics as opposed to the ordinary spherical harmonics reflects the nontrivial topological structure of massless particles with nonzero helicity. The resultant angular momentum multiplet structure is quite different than that of massive particles, with at most one multiplet for each angular momentum $j$ and with $|h|$ acting as a lower bound on $j$. This illustrates the obstruction to a spin-orbital decomposition of the angular momentum for massless particles, as such a sparse multiplet structure is not consistent with any reasonable spin-orbital splitting.}

\begin{document}
\maketitle
\flushbottom

\section{Introduction}\label{sec:SWSH:Intro}
After the breakthrough discovery of topological insulators \cite{Kane2005,Hasan2010}, it has been realized that nontrivial topological structures can also exist in continuous media such as fluids  \cite{Souslov2017,Delplace2017,Tauber2019,Souslov2019,Perrot2019,Parker2021,Faure2023} and plasmas \cite{Yang2016,Gao2016,Parker2020a,Parker2020b,Parker2021,Fu2021,Fu2022,Qin2023, Qin2024plasma,Fu2024}. In fact, even the vacuum supports topologically nontrivial particle waves \cite{PalmerducaQin_PT,PalmerducaQin_GT, PalmerducaQin_helicity}. The topology in the latter refers to the global relationship between the internal (polarization) and external (momentum) degrees of freedom (DOFs) of a particle. For a massive particle, the internal spin DOFs are independent of the external momentum DOFs, reflecting the topological triviality of massive particles. In contrast, the internal polarization states of a massless particle with nonzero helicity are constrained by the external momentum DOFs, for example, the photon polarization must be transverse to the momentum. This ``twisting'' together of the internal and external DOFs of a massless particle reflects a nontrivial topological structure. More technically, for each momentum $k$, the possible polarization states form a vector space $V(k)$ which varies smoothly with $k$. Such a parameterized set of vector spaces form a topological structure known as a vector bundle, and these can be broadly classified as trivial or nontrivial, or more narrowly classified by their topological invariants such as Chern numbers \cite{Tu2017differential}. 

Recent work has leveraged this topological structure to discover new properties of massless particles. For example, we explicitly constructed a globally smooth basis of polarization vectors for photons and gravitons \cite{PalmerducaQin_PT,PalmerducaQin_GT}, despite the conventional assertion that the existence of such a global basis would violate the hairy ball theorem \cite{Tong2006,Woit2017}. We also proved that it is not possible to decompose the total angular momentum operator $\boldsymbol{J}$ for massless particles into legitimate spin (SAM) and orbital angular momentum (OAM) operators \cite{PalmerducaQin_SAMOAM, PalmerducaQin_connection}, generalizing the results of van Enk and Nienhuis \cite{VanEnk1994_EPL_1,VanEnk1994_JMO_2}. 

Recently, Dragon \cite{Dragon2024} used vector bundle techniques to construct simultaneous eigenstates of energy, helicity, $J^2$ and $J_z$ for massless particles. This gives a decomposition of the state space into finite-dimensional rotationally invariant subspaces and gives a countable orthonormal basis for monochromatic waves (as opposed to the uncountable basis of momentum eigenstates). The analogous states for massive particles, with helicity replaced by spin, are simply given as Clebsch--Gordan sums of tensor products of the spherical harmonics with definite spin states. The problem is much less trivial for massless particles where the internal (polarization) and external (momentum) DOFs are coupled, but the problem is tractable in Dragon's vector bundle framework. However, the form of these solutions have certain undesirable properties. One can describe vector bundles with a number of different formalisms. In that used by Dragon, the fibers are never explicitly defined; instead the vector bundle structure is probed solely through transition functions on overlapping domains of the north and south stereographic wavefunctions. While this is valid mathematically, it leads to eigenstates which are represented as pairs of functions in stereographic coordinates which are not easily interpreted. Furthermore, Dragon only gives the explicit formula for a single eigenstate within each angular momentum multiplet; to obtain the others one would need to repeatedly apply the angular momentum lowering operator, which itself has a somewhat complicated form (cf. \cite{Dragon2024}, Eq. (69)). 

In contrast to Dragon's implicit framework, we recently showed that massless particles of arbitrary helicity can be explicitly represented as iterated tensor products of the right ($R$) and left ($L$) photon bundles, with these photon bundles taking the form of line subbundles of the complexified tangent bundle of the sphere \cite{PalmerducaQin_helicity}. In this article, we use this alternative framework to give a simpler and more explicit derivation of the simultaneous eigenstates of energy, helicity, $J^2$, and $J_z$ for massless particles. We show that in appropriate coordinates, these eigenstates actually take the form spin-weighted spherical harmonics (SWSHs), generalizing the ordinary spherical harmonics that appear in the study of massive particles. SWSHs have shown up in many physical applications, including the study of asymptotically flat spacetimes \cite{Newman1966}, the dynamics of charged particles in the presence of a magnetic monopole \cite{Wu1975,Wu1976,Dray1985,Fakhri2007}, the study of complex spacetimes and twistor theory \cite{Curtis1978}, in geophysical \cite{Michel2020} and computer graphics applications \cite{Yi2024}, and in studying anisotropies in the cosmic microwave background \cite{Zaldarriaga1997,ng1999,Wiaux2006}. They have been used in treatments of electromagnetic and gravitational waves in position space where the theory is more complicated and the topological nature of the SWSHs is not as apparent \cite{Thorne1980,Torres2007,Mandrilli2020}. Here, we show that when working in momentum space, the SWSHs give a simple description of the angular momentum eigenfunctions for massless particles of arbitrary helicity. 

Another important implication of the topological nontriviality of massless particles relates to the question of whether or not the total angular momentum operator can be split into SAM and OAM operators. For massive particles such a splitting is trivial in the nonrelativistic limit and given by the Newton-Wigner splitting in the relativistic case \cite{Terno2003}. In contrast, there has been a long controversy surrounding this question for massless particles \cite{Akhiezer1965,Jaffe1990,VanEnk1994_EPL_1,VanEnk1994_JMO_2,Chen2008,Wakamatsu2010,Bialynicki-Birula2011,Leader2013,Leader2014,Leader2016,Leader2019,Yang2022,Das2024}, with many different proposed massless spin-orbital decompositions; however, all fail to satisfy the defining property of angular momentum operators. In particular, none of the proposed operators generate 3D rotational symmetries (or equivalently, satisfy the angular momentum commutation relations) while remaining gauge invariant \cite{VanEnk1994_EPL_1,VanEnk1994_JMO_2,Leader2014,Leader2019,Yang2022}. We recently proved that the topological nontriviality of massless particles actually makes such a spin-orbital decomposition impossible \cite{PalmerducaQin_SAMOAM}, essentially showing that the issues first recognized by van Enk and Nienhuis \cite{VanEnk1994_EPL_1} cannot be avoided by redefining the SAM and OAM operators. The physics of this singularity can be traced back to the vanishing of the rest frame \cite{PalmerducaQin_PT}, the degeneracy of transverse boosts and rotations that occurs in the massless limit \cite{PalmerducaQin_connection}, or to the geometric jump in the little group from $\SO(3)$ to $\ISO(2)$ in the massless limit \cite{PalmerducaQin_PT,PalmerducaQin_GT,PalmerducaQin_SAMOAM}. Nevertheless, the proofs of the nonexistence of massless SAM and OAM are somewhat abstract \cite{PalmerducaQin_SAMOAM}. The second goal of this article is to show that the description of the rotationally invariant subspaces in terms of SWSHs gives a concrete illustration of this no-go result. In particular, we show that the angular momentum multiplet structure of the SWSHs is too sparse to be consistent with any reasonable spin-orbital decomposition. Furthermore, we also find that the helicity acts as lower bound on the total angular momentum of a massless particles. In contrast, massive particles possess states of total angular momentum $0$, regardless of the spin, and such states are unavoidable if the the total OAM takes on all nonnegative integer values. 

We will ignore quantum numbers arising from non-spacetime symmetries, so our results will not strictly apply to particles with, for example, color charge. We also restrict our attention to massless bosons since this leads to simplifications and there are no known elementary massless fermions. We also note that an alternative form of the simultaneous eigenfunctions of energy, helicity, and angular momentum in terms of Wigner $D$-functions and $3j$ symbols can be found in Ref. \cite{LandauLifshitzQED}. However, the derivation is more complicated and the topological nature of these eigenstates is not clear in this form of the solution. Furthermore, the relationship between these eigenstates and the problem of SAM-OAM decomposition has not been previously considered.

This article is organized as follows. In Sec. \ref{sec:SWSH:Formalism} we give an overview of the vector bundle representations of massless particles. We then show how massless particles states are described by spin-weighted functions in Sec. \ref{sec:SWSH:Spin_weighted_functions} . In Sec. \ref{sec:SWSH:bundle_harmonics} we construct the total angular momentum operators for massless particles and show that they decompose the state space into angular momentum multiplets corresponding to the SWSHs. In Sec. \ref{sec:SWSH:SAM_OAM} we show this decomposition illustrates the topological obstruction to a spin-orbital splitting of the massless angular momentum operator.

\section{Vector bundle representations of particles}\label{sec:SWSH:Formalism}
We will use the vector bundle formalism for massless particles developed in Refs. \cite{PalmerducaQin_PT,PalmerducaQin_GT,PalmerducaQin_helicity, PalmerducaQin_SAMOAM}, which we now review. The momentum four vector $k^\mu$ of a massless particle resides on the forward light cone $\Lightcone$, which, in the $(-,+,+,+)$ signature, is given by
\begin{align}
    \Lightcone &= \{k^\mu= (\omega, \kvec): k^\mu k_\mu = 0, \omega > 0 \} \\
    &\cong \pspace.
\end{align}
As massless particles cannot have zero momentum, there is a hole at the origin, making the momentum space noncontractible. Note that this light cone can be parameterized solely by its spatial part $\kvec$, giving a diffeomorphism with the punctured Euclidean space $\pspace$. In the prototypical case of photons, any definite momentum state can be represented by $(\kvec,\mbf{E})$ or $(k^\mu,\mbf{E})$ where $\mbf{E}$ is the electric field polarization which must be transverse to the momentum: $\mbf{E} \cdot \kvec = 0$. These are representations in Fourier space, so $\mbf{E}$ is generally complex. At each $\kvec \in \Lightcone$, the collection of all such transverse $\mbf{E}$ forms a vector space $V(\kvec)$. Such a parameterized collection of vector spaces forms a vector bundle $\gamma$, which we call the total photon bundle, consisting of all $(\kvec,\mbf{E})$ pairs. $\Lightcone$ is referred to as the base manifold, and there is a projection $\pi:\gamma \rightarrow \Lightcone$ such that $\pi(\kvec,\mbf{E}) = \kvec$. The vector space $V(\kvec)$ is called the fiber at $\kvec$. The rank of the vector bundle is the dimension of any fiber, describing the number of internal DOFs; the rank of $\gamma$ is $2$. 

The photon bundle forms a representation of the (proper orthochronous) \Poincare group $\ISO^+(3,1)$: for a Lorentz transformation $\Lambda \in \SO^+(3,1)$ we have 
\begin{equation}\label{eq:SWSH:bundle_action}
    \Lambda (k^\mu,\mbf{E}) = (\Lambda k^\mu,\mbf{E}')
\end{equation}
where $k^\mu$ transforms like a four-vector and $\mbf{E}'$ transforms under rotations like a three-vector and under boosts according to the standard electromagnetic boost transformations \cite{Jackson1999,PalmerducaQin_PT}. Under a spacetime translation $T_a$ by $a^\mu \in \Real^4$, we have
\begin{equation}
    T_a (k^\mu,\mbf{E}) = e^{-ik^\mu a_\mu}(k^\mu,\mbf{E}).
\end{equation}
Elementary particles correspond to unitary irreducible representations (UIR) of the \Poincare group \cite{Weinberg1995,Maggiore2005}, so we decompose $\gamma$ into its Lorentz irreducible subbundles. These are the line (rank 1) bundles $\gamma_+$ and $\gamma_-$ consisting the right (R) and left (L) circularly polarized photons, respectively. That is, $\gamma_\pm(\kvec)$ consists of vectors of the form 
\begin{equation}\label{eq:SWSH:explicit_RL_form}
    c(\mbf{E}_1 \pm i\mbf{E}_2)
\end{equation}
 where $c \in \Comp$ and $(\mbf{E}_1,\mbf{E}_2,\khat)$ forms a right-handed orthonormal basis of $\Real^3$.

Via Wigner's classification \cite{Wigner1939,Weinberg1995,Asorey1985,PalmerducaQin_PT}, massless bosons are characterized by an integer helicity $h$. It was recently shown that massless particle bundles of arbitrary helicity can be explicitly constructed by taking repeated tensor products of the $R$ and $L$ photon bundles:
\begin{equation}\label{eq:SWSH:gamma_h_tensor}
    \gamma_h = \underbrace{\gamma_\pm \otimes \cdots \otimes \gamma_\pm}_{|h| \text{ times}}
\end{equation}
where the $\pm$ corresponds to the sign of $h$; note that this is a distinctly massless result that relies on the fact that massless particles have a single internal DOF. We can write a vector in $\gamma_h$ as $v = v_1 \otimes \cdots \otimes v_h$ with $v_n \in \gamma_\pm$, and the \Poincare action on $\gamma_h$ is given by
\begin{align}\label{eq:SWSH:bundle_action_tensor}
    \Lambda(\kvec,v_1 \otimes \cdots \otimes v_n) &= (\Lambda\kvec,\Lambda v_1 \otimes \cdots \otimes \Lambda v_n), \\
    T_a(\kvec,v_1 \otimes \cdots \otimes v_n) &= e^{-ik^\mu a_\mu}(\kvec,v_1 \otimes \cdots \otimes v_n).
\end{align}
Eqs. (\ref{eq:SWSH:bundle_action}) and (\ref{eq:SWSH:bundle_action_tensor}) describe the Lorentz transformation on single-particle states with definite momentum $\kvec$. The Hilbert space of particle states is $L^2(\gamma_h)$, the collection of square integrable sections of $\gamma_h$ with respect to the Lorentz invariant measure. An element $\mbf{\alpha} \in L^2(\gamma_h)$ can be considered as a wave function, where $\mbf{\alpha}(\kvec) \in \gamma_h(\kvec)$. The inner product is given by
\begin{align}
    \langle\mbf{\alpha}, \mbf{\beta}\rangle &= \int_{\Lightcone}\frac{d^3\kvec}{|\kvec|}\mbf{\alpha}^*(\kvec)\cdot \mbf{\beta}(\kvec) \\
    &= \prod_{n=1}^{|h|} \int_{\Lightcone}\frac{d^3\kvec}{|\kvec|}\mbf{\alpha}^*_n(\kvec)\cdot \mbf{\beta}_n(\kvec).
\end{align}
The \Poincare action on $\gamma_h$ induces a representation of the \Poincare group on $L^2(\gamma_h)$ given by 
\begin{equation}
    (\Lambda\mbf{\alpha})(\kvec) = \Lambda\Big(\mbf{\alpha}(\Lambda^{-1}\kvec)\Big)
\end{equation}
for $\Lambda \in \ISO^+(3,1)$. We are particularly concerned with the subgroup $\SO(3) \seq \ISO^+(3,1)$ corresponding to rotations of the spacelike coordinates. The infinitesimal generators of this $\SO(3)$ action are the angular momentum operators $\boldsymbol{J}$. That is, for some direction $\vhat \in \Real^3$, the scalar operator $\vhat \cdot \boldsymbol{J}$ is defined by
\begin{equation}\label{eq:SWSH:J_def}
   [(\vhat \cdot \mbf{J})\mbf{\alpha}](\kvec) = i \frac{d}{d\psi}\Big |_{\psi=0}\mbf{R}_{\vhat}^\psi \mbf{\alpha}\big(\mbf{R}_{\vhat}^{-\psi} \kvec\big)
\end{equation}
where $\mbf{R}_{\vhat}^\psi$ is a rotation by $\psi$ about $\vhat$. Similarly, the momentum operator $\mbf{P}$ and Hamiltonian $H$ are the generators of spacetime translations and act simply by
\begin{align}
   \mbf{P}\mbf{\alpha}(\kvec) &= \kvec \mbf{\alpha}(\kvec), \\
   H\mbf{\alpha}(\kvec) &= |\kvec|\mbf{\alpha}(\kvec).
\end{align}

\section{Massless particle states as spin-weighted functions}\label{sec:SWSH:Spin_weighted_functions}
In this section we will review the formalism of spin-weighted functions and show that massless particle states of helicity $h$ can be described using these functions if appropriate coordinates are chosen. One of the standard definitions of spin-weighted functions is as follows \cite{Goldberg1967,Dray1985,Boyle2016}. Denote the points on the 2-sphere $S^2$ by the normal unit vector $\khat$. Suppose that $(\mbf{a},\mbf{b},\khat)$ is a right-handed orthonormal basis at each $\khat$. Note that by the hairy ball theorem $\mbf{a}$ and $\mbf{b}$ will have singularities at at least one $\khat$; we will consider $\mbf{a},\mbf{b}$ such that they vary smoothly with $\khat$ except at a finite number of points. We can encode the choice of $\mbf{a}$ and $\mbf{b}$ in the complex vector
\begin{align}\label{eq:SWSH:m_plus_def}
    \mbf{m}_+ = \frac{\mbf{a} + i\mbf{b}}{\sqrt2}.
\end{align}
A different choice of orthonormal frame $(\mbf{a}',\mbf{b}',\khat)$ corresponds to a multiplication of $\mbf{m}_+$:
\begin{gather}
    \mbf{m}_+' = e^{i\xi(\khat)}\mbf{m}_+ \\
    \mbf{a}' = \cos(\xi)\mbf{a} - \sin(\xi)\mbf{b} \\
    \mbf{b}' = \sin(\xi)\mbf{a} + \cos(\xi)\mbf{b}
\end{gather}
for some $\xi:S^2 \rightarrow \Real$. We consider $\Comp$-valued functions $f(\mbf{m}_+,\khat)$ on the sphere that depend not only on $\khat \in S^2$, but on the choice of right-handed orthonormal basis defined by $\mbf{m}_+$. Such a function is said to have spin-weight $s$ and is written as $_sf$ if
\begin{equation}
    _sf(e^{i\xi}\mbf{m}_+,\khat) = e^{is\xi} {_sf}(\mbf{m}_+,\khat)
\end{equation}
for any $\xi$. We note that it is straightforward to see that spin-weighted functions can equivalently be defined in terms of 
\begin{equation}
    \mbf{m}_- = \frac{\mbf{a} - i\mbf{b}}{\sqrt{2}}.
\end{equation}
The transformation $\mbf{m}_+ \rightarrow e^{i\xi}\mbf{m}_+\doteq\mbf{m}'_+$ corresponds to $\mbf{m}_-\rightarrow e^{-i\xi}\mbf{m}_-\doteq\mbf{m}'_-$, that is, these both describe $(\mbf{a},\mbf{b},\khat) \rightarrow(\mbf{a}',\mbf{b}',\khat)$ . Then a function $_sf$ which depends on $\khat$ and a choice of orthonormal basis has spin-weight $s$ if
\begin{equation}
    _s f(e^{-i\xi}\mbf{m}_-,\khat) = e^{is\xi}{_sf}(\mbf{m}_-,\khat).
\end{equation}
We will now show how the state space $L^2(\gamma_h)$ of a massless helicity $h$ particles corresponds to the space of spin-weight $-h$ functions. $\mbf{m}_+(\khat) \in \gamma_+(\kvec)$ by (\ref{eq:SWSH:explicit_RL_form}) and (\ref{eq:SWSH:m_plus_def}), and since $\gamma_+$ is a line bundle, $\mbf{m}_+$ forms a basis of the fiber $\gamma_+(\khat)$. Therefore, we can express any section  $ \mbf{\alpha} \in L^2(\gamma_+)$ by 
\begin{align}
    \mbf{\alpha}(\kvec) &= a(\kvec)\mbf{m}_+(\khat) = \big(e^{-i\xi}a(\kvec)\big)\big(e^{i\xi}\mbf{m}_+(\khat)\big) \\
    &= \big(e^{-i\xi}a(\kvec)\big)\mbf{m}_+'(\khat)
\end{align}
We can thus define a basis dependent function $_{-1}a(\mbf{m}_+,\kvec)$ such that
\begin{equation}
    \mbf{\alpha}(\kvec) = {_{-1}a}(\mbf{m}_+,\kvec)\mbf{m}_+
\end{equation}
for any choice of $\mbf{m}_+$. We can express $\kvec = |\kvec|\khat \doteq \omega \khat$, and write $_{-1}a$ as ${_{-1}a}(\mbf{m}_+,\khat,\omega)$. Then for any fixed $\omega_0$, ${_{-1}\alpha}(\mbf{m}_+,\khat,\omega_0)$ is a function of spin-weight $-1$. Therefore, we see that $R$ photon states $L^2(\gamma_+)$ are described by spin-weight $-1$ functions. The spin-weighting encodes the local basis dependence of any description of $\mbf{\alpha}$. By an analogous argument with $\mbf{m}_+$ replaced by $\mbf{m}_-$, we find that the $L$-photon states, $L^2(\gamma_-)$, are described by spin-weight $1$ functions. We can extend these results to bundles of higher helicity using Eq. (\ref{eq:SWSH:gamma_h_tensor}). A section $\mbf{\beta} \in L^2(\gamma_h)$ can be written as
\begin{align}
    \mbf{\beta}(\mbf{k}) &= b(\kvec) \bigotimes_{n=1}^{|h|}\mbf{m}_{\pm} \\
    &= e^{\mp i|h|\xi}b(\kvec) \bigotimes_{n=1}^{|h|}e^{\pm i\xi}\mbf{m}_{\pm} \\
    &= e^{-ih\xi}b(\kvec) \bigotimes_{n=1}^{|h|}\mbf{m}'_{\pm}
\end{align}
where the $\pm$ sign is corresponds to the sign of $h$. Thus, massless particle states of helicity $h$ are described by function of spin-weight $-h$, and we see that this is essentially a different (and arguably bulkier) way to describe the fact that $\mbf{\beta}$ is a section of the vector bundle $\gamma_h$. The relationship between massless particle states and spin-weighted functions helps explain why SWSHs appear in the study of the angular momentum eigenstates which we examine in the next section. 

\section{Expressing angular momentum eigenstates as spin-weighted spherical harmonics}\label{sec:SWSH:bundle_harmonics}
In this section we will show that the angular momentum eigenstates of massless particles are described by the SWSHs, provided we choose our coordinates appropriately. We can peel off the radial behavior in $\kvec$ space since it is unaffected by the angular momentum $\mbf{J}$. Denote by $\gamma_{h,S^2}$ the bundle $\gamma_h$ with base manifold restricted to the unit sphere $S^2$ in $\mbf{k}$ space. Then
\begin{equation}
    L^2(\gamma_h) = L^2(\Real^+)\otimes L^2(\gamma_{h,S^2})
\end{equation}
where $\Real^+$ denotes the positive real numbers and $L^2(\gamma_{h,S^2})$ represents the space of monochromatic waves \footnote{More explicitly, we can write $L^2(\Real^+,|\kvec|^2d|\kvec|)$ to emphasize that the radial $L^2$ space has measure $|\kvec|^2d|\kvec|$ (\cite{Reed1975}, p. 160).}. Note that this decomposition can be considered as induced by the observables $H$ and $\mbf{J}$; $H$ acts trivially on $L^2(\gamma_{h,S^2})$ while $\mbf{J}$ acts trivially on $L^2(\Real^+)$. Our goal then is to find the simultaneous eigenstates of $J^2$ and $J_z$ in $L^2(\gamma_{h,S^2})$, which is equivalent to decomposing $L^2(\gamma_{h,S^2})$ into UIRs of $\SO(3)$. 

\subsection{Decomposing the angular momentum}
It will be useful to decompose $\mbf{J}$ as
\begin{subequations}\label{eq:SWSH:par_perp_decomp}
\begin{align}
    \mbf{J} &= \mbf{J}_\parallel + \mbf{J}_\perp \\
    \mbf{J}_\parallel &\doteq (\Phat \cdot \mbf{J})\Phat \\
    \mbf{J}_\perp &\doteq \mbf{J} - (\Phat \cdot \mbf{J})\Phat.
\end{align}
\end{subequations}
Given that $\mbf{J}_\parallel$ is related to the helicity operator $\chi \doteq \Phat \cdot \mbf{J}$, it is tempting to call $\mbf{J}_\parallel$ and $\mbf{J}_\perp$ SAM and OAM operators. We emphasize that while these are well-defined vector operators, meaning that \cite{Hall2013}
\begin{subequations}
\begin{align}
    [J_{\parallel,a},J_b] &= i\epsilon_{abc}J_{\parallel,c} \\
    [J_{\perp,a},J_b] &= i\epsilon_{abc}J_{\perp,c},
\end{align}
\end{subequations}
they do not satisfy the angular momentum commutation relations,
\begin{subequations}
\begin{align}
    [J_{\parallel,a},J_{\parallel,b}] &\neq i\epsilon_{abc}J_{\parallel,c} \\
    [J_{\perp,a},J_{\perp,b}] &\neq i\epsilon_{abc}J_{\perp,c},
\end{align}
\end{subequations}
and are thus not true angular momentum operators. This was first discovered by van Enk and Neinhius \cite{VanEnk1994_EPL_1,VanEnk1994_JMO_2}. Nevertheless, it was recently shown that this is the unique massless splitting induced by \Poincare symmetry and is essentially the closest thing to a spin-orbital decomposition that exists for massless particles \cite{PalmerducaQin_connection}. It is noteworthy that $\mbf{J}_\parallel$ is a point operator \cite{Tu2017differential, PalmerducaQin_connection}, meaning that $(\mbf{J}_\parallel \mbf{\alpha})(\kvec_0)$ depends only on $\mbf{\alpha}(\kvec_0)$ and not on the value of $\mbf{\alpha}$ at any other $\kvec$. In physical terms, this means that $\mbf{J}_\parallel$ only generates transformations of the internal DOFs. In contrast, $\mbf{J}_\perp$ is a local differential operator, with $(\mbf{J}_\perp \mbf{\alpha})(\kvec_0)$ depending on the value of $\mbf{\alpha}(\kvec)$ in a neighborhood of $\kvec_0$.

It is of both theoretical and practical interest to relate $\mbf{J}_\parallel$ and $\mbf{J}_\perp$ to the more familiar SAM and OAM $\mbf{S}$ and $\mbf{L}$ of massive particles. The fibers of $\gamma_{\pm}(\kvec)$ are embedded in the Hermitian vector space $\Comp^3$ via Eq. (\ref{eq:SWSH:explicit_RL_form}); let $\proj:\Lightcone \times \Comp^3 \rightarrow \gamma_\pm$ be the projection onto the photon states. Then for a section $\mbf{\alpha} \in L^2(\gamma_{\pm})$,
\begin{align}
    [(\vhat \cdot \mbf{J})\mbf{\alpha}](\kvec) &= i \proj\frac{d}{d\psi}\Big |_{\psi=0}\mbf{R}_{\vhat}^\psi \mbf{\alpha}\big(\mbf{R}_{\vhat}^{-\psi} \kvec\big) \\
    &= i \proj \Big\{ \frac{d}{d\psi}\Big |_{\psi=0}\mbf{R}_{\vhat}^\psi \mbf{\alpha}(\kvec) +  \frac{d}{d\psi}\Big |_{\psi=0} \mbf{\alpha}\big(\mbf{R}_{\vhat}^{-\psi} \kvec\big) \Big\} \\
    &= \vhat \cdot (\proj\circ \boldsymbol{S}) \mbf{\alpha}(\kvec) + \vhat \cdot (\proj\circ \boldsymbol{L})\mbf{\alpha}(\kvec). \label{eq:SWSH:S_L_decomp}
\end{align}
Here,
\begin{align}
    (S_a)_{bc} = -i\epsilon_{abc}
\end{align}
and in spherical coordinates \cite{Zettili2009}
\begin{equation}
    \mbf{L} = -i\ephi \partial_\theta - \etheta \frac{1}{\sin\theta} \partial_\phi.
\end{equation}
These extend to operators on the tensor product bundles $\gamma_h$ via the product rule, that is,
\begin{equation}\label{eq:SWSH:product_rule_extension}
    (\proj \circ \mbf{S})(\mbf{\alpha}_1 \otimes \cdots \otimes \mbf{\alpha}_{|h|}) \doteq \sum_{m=1}^{|h|} \mbf{\alpha}_1 \otimes \cdots \otimes(\proj \circ \mbf{S})\mbf{\alpha}_m \otimes \cdots \otimes \mbf{\alpha}_{|h|}
\end{equation}
and similarly for $\proj \circ \mbf{L}$.
The following lemma states that the projected operators $\proj \circ \mbf{S}$ and $\proj \circ \mbf{L}$ are in fact the same as $\mbf{J}_\parallel$ and $\mbf{J}_\perp$.
\begin{lemma}
    $\mbf{J}_\parallel = \proj \circ \mbf{S}$ and $\mbf{J}_\perp = \proj \circ \mbf{L}$. In particular,
    \begin{equation}\label{eq:SWSH:prime_decomp}
        \mbf{J} = (\Phat \cdot \mbf{J})\Phat + \proj \circ \mbf{L}
    \end{equation}
\end{lemma}
\begin{proof}
    $\proj \circ \mbf{S}$ restricts to an operator on each fiber of $\gamma_h$. Suppose first that $h = \pm 1$ so that $\mbf{\alpha}(\kvec)$ is embedded in $\Comp^3$.  Let $\vhat \in \mathbb{R}^3$ and consider the fiber at $\khat$, with $\khat \cdot \vhat = \cos(\theta)$. Decompose $\vhat = \cos(\theta)\khat + \vhat_\perp$ into parallel and perpendicular components with respect to $\khat$. Then
    \begin{align}
    [\vhat \cdot (\proj \circ \mbf{S})]\mbf{\alpha}(\kvec)  &= \Big( \cos(\theta)\khat\cdot (\proj \circ \mbf{S}) + \vhat_\perp \cdot (\proj \circ \mbf{S}) \Big)\mbf{\alpha}(\kvec).
    \end{align}
    Using the identity
    \begin{equation}
        \frac{d}{d\psi}\Big |_{\psi = 0}R_{\vhat_1}^\psi \mbf{v}_2 = \vhat_1 \times \mbf{v}_2
    \end{equation}
    we find that
    \begin{align}\label{eq:SWSH:perp_part}
        [\vhat_\perp \cdot (\proj \circ \mbf{S})] \mbf{\alpha}(\khat) &= i\proj \frac{d}{d\psi}\Big |_{\psi = 0}R_{\vhat_\perp}^\psi \mbf{\alpha}(\kvec) \\&= i\proj (\mbf{v}_\perp \times\mbf{\alpha}) \\
        &= 0.
    \end{align} 
    The last equality follows since $\vhat_{\perp} \times\alpha(\mbf{k})$ is proportional to $\khat$ and is thus annihilated by $\proj$. $\proj(\khat \cdot \mbf{S}) = \khat \cdot \mbf{S}$ since $R_{\khat}^\psi \mbf{\alpha}(\kvec) \cdot \khat = 0$, so
    \begin{align}
        [\vhat \cdot (\proj \circ \mbf{S})]\mbf{\alpha}(\kvec) &= \cos(\theta) \proj(\khat \cdot \mbf{S})\mbf{\alpha}(\kvec) \\
        &= (\vhat \cdot \khat)(\khat \cdot \mbf{S})\bm\alpha(\kvec) \\
        &= i(\vhat \cdot \khat)\frac{d}{d\psi}\Big |_ {\psi = 0}R_{\khat}^{\psi} \mbf{\alpha}(\kvec) \\
        &= i(\vhat \cdot \khat)\frac{d}{d\psi}\Big |_ {\psi = 0}R_{\khat}^{\psi} \mbf{\alpha}(R_{\khat}^{-\psi}\kvec) \\
        &= (\vhat \cdot \khat)(\khat \cdot \mbf{J})\mbf{\alpha}(\kvec) \\
        &= (\vhat \cdot \mbf{J}_\parallel)\mbf{\alpha}(\kvec),
    \end{align}
    proving that $\mbf{J}_\parallel = (\proj \circ \mbf{S})$. It then follows from Eqs. (\ref{eq:SWSH:par_perp_decomp}) and (\ref{eq:SWSH:S_L_decomp}) that $\mbf{J}_\perp = (\proj \circ \mbf{L})$, proving the lemma for $h = \pm1$. The extension to arbitrary $h$ follows from Eq. (\ref{eq:SWSH:product_rule_extension}).
\end{proof}
$\mbf{J}_\parallel = \proj \circ \mbf{S}$ and $\mbf{J}_\perp = \proj \circ \mbf{L}$ are well-defined operators on the Hilbert space of sections $L^2(\gamma_h)$, however, the same is not true of $\mbf{S}$ and $\mbf{L}$.  Indeed, $(\mbf{S}\mbf{\alpha})(\kvec)$ and $(\mbf{L}\mbf{\alpha})(\kvec)$ are not generally orthogonal to $\kvec$ and thus are not sections of $\gamma_{h}$.
While $\mbf{S}$ and $\mbf{L}$ do each satisfy $\so(3)$ commutation relations, 
\begin{align}
    [S_a,S_b] = i\epsilon_{abc}S_c, \\
    [L_a,L_b] = i\epsilon_{abc}L_b,
\end{align}
composition with the projection operator destroys this rotational symmetry. The projected operators instead satisfy nonstandard commutation relations 
\begin{align}
    [J_{\parallel,a} , J_{\parallel,b}] &= 0, \\
    [J_{\perp,a},J_{\parallel,b}] &= i\epsilon_{abc}J_{\parallel,c}, \\
    [J_{\perp,a},J_{\perp,b}] &= i\epsilon_{abc}(J_{\perp,c} - J_{\parallel,c}),
\end{align}
with $\mbf{J}_\parallel$ generating a 3D translational symmetry while the vector operator $\mbf{J}_\perp$ is not the generator of any symmetry group at all (since it is not closed under commutations).

\subsection{Solving for angular momentum eigenstates}
We now decompose the angular Hilbert space $L^2(\gamma_{h,S^2})$ into UIRs of $\SO(3)$, which is equivalent to solving for the angular momentum multiplets of sections $(_h\mbf{\alpha}_{jm})$ with $-j \leq m \leq j$ which satisfy
\begin{subequations}
\label{eq:SWSH:AM_bundle}
\begin{align}
    J_z \, _h\mbf{\alpha}_{jm} &= m \,{_h\mbf{\alpha}_{jm}}, \\
    J_\pm \, _h\mbf{\alpha}_{jm} &= [(j \mp m)(j + 1 \pm m)]^{1/2} \,  {_h\mbf{\alpha}_{jm\pm 1}}, \label{eq:SWSH:ladder}\\
    J^2 \, _h\mbf{\alpha}_{jm} &= j(j+1) \, _h\mbf{\alpha}_{jm},
\end{align}
\end{subequations}
where $J_\pm = J_x \pm iJ_y$. To proceed, we will express the $\mbf{J}$ operators in a particular choice of basis for $\gamma_{h,S^2}$ which will elucidate the relationship between the total angular momentum states and the spin-weighted spherical harmonics. If we choose the standard spherical unit vectors $(\etheta,\ephi)$ then $\epm \doteq \frac{1}{\sqrt{2}}(\etheta \pm i \ephi)$ forms a smooth basis of $\gamma_\pm$ at all $\khat \in S^2$ except the poles; similarly,
\begin{equation}
    \eh \doteq 2^{-\frac{|h|}{2}}\underbrace{(\etheta \pm i \ephi) \otimes \cdots \otimes (\etheta \pm i \ephi)}_{|h| \text{ times}}
\end{equation}
forms such a basis of $\gamma_h$, where the $\pm$ signs correspond to the sign of $h$. A globally smooth section $\mbf{\alpha}$ of $\gamma_{h,S^2}$ is expressed as
\begin{equation}
    \mbf{\alpha}(\theta, \phi) = f(\theta, \phi)\eh.
\end{equation}
$\mbf{\alpha}$  and $\eh$ are smooth for $\theta \neq 0,\pi$, so the same is true of $f(\theta,\phi)$. Since $\mbf{\alpha}$ is also smooth at the poles, the limits
\begin{subequations}
\begin{align}
    \lim_{\theta \rightarrow 0}f(\theta,\phi)\eh(\theta,\phi) &= \lim_{\theta \rightarrow 0}f(\theta,\phi)e^{-ih \phi}(\ex \pm i \ey)^{|h|} \\
    \lim_{\theta \rightarrow \pi}f(\theta,\phi)\eh(\theta,\phi) &= \lim_{\theta \rightarrow \pi}f(\theta,\phi)e^{ih \phi}(\ex \pm i \ey)^{|h|}
\end{align}
\end{subequations}
must exist and agree for all $\phi$. This implies that asymptotically
\begin{subequations}
\label{eq:SWSH:limiting_behavior}
\begin{align}
    f(\theta, \phi) &\sim c_N e^{ih\phi},\;\; (\theta \rightarrow 0^+)\\
    f(\theta, \phi) &\sim c_S e^{-ih\phi},\;\; (\theta \rightarrow \pi^-)
\end{align}
\end{subequations}
for constants $c_N$ and $c_S$. If $c_N$ or $c_S$ are nonzero, then $f(\theta,\phi)$ is not smooth at $\theta = 0$ or $\theta = \pi$, respectively. However, we see that such cases correspond only to coordinate singularities, and not actual singularities as the state $\boldsymbol{\alpha}$ is still globally smooth.

To solve for the angular momentum multiplets, we express $J_z$,   $J_\pm$, and $J^2$ in the $\eh$ frame so that
\begin{subequations}
\label{eq:SWSH:multiplet_bundle}
\begin{align}
J_z (f \eh) &= J'_z (f) \eh, \label{eq:SWSH:J_z_alpha}\\
J_\pm (f \eh) &= J'_\pm (f) \eh, \label{eq:SWSH:J_pm_alpha}\\
J^2(f \eh) &= J'^2 (f) \eh \label{eq:SWSH:J^2_alpha},
\end{align}
\end{subequations}
for operators $J'_z$, $J'_\pm$, $J'^2$ which will be determined. By the definition of $\mbf{J}$,
\begin{equation}
    \mbf{J}(f\eh) = \mbf{L}(f)\eh + f \mbf{J}(\eh).
\end{equation}
Since $\eh$ is rotationally invariant about the $z$-axis, it follows from Eq. (\ref{eq:SWSH:J_def}) that
\begin{equation}
    J'_z = -i \partial_\phi.
\end{equation}
To find $J'_\pm$ and $J'^2$ we will use the decomposition in Eq. (\ref{eq:SWSH:prime_decomp}). We have
\begin{align}
    \mbf{J}_\perp &= -i \Big( \ephi \proj \circ \partial_\theta - \frac{\etheta}{\sin{\theta}} \proj \circ \partial_{\phi} \Big) \label{eq:SWSH:J_perp}
\end{align}
where $\proj$ annihilates any $\khat$ components produced by the angular derivatives. Direct calculation gives
\begin{align}
    J_{\perp,x}\eh &= h \frac{\cos^2 \theta}{\sin{\theta}} \cos \phi \, \eh \label{eq:SWSH:J_perp_x}\\
    J_{\perp,y}\eh &= h \frac{\cos^2 \theta}{\sin{\theta}} \sin \phi \, \eh \label{eq:SWSH:J_perp_y}\\
    J_{\perp,z}\eh &= -h \cos \theta \, \eh  \label{eq:SWSH:J_perp_z}\\
    (\mbf{v} \cdot \mbf{J}_{\parallel})\eh &= (\mbf{v} \cdot \khat)h \eh. \label{eq:SWSH:J_par}
\end{align}
Defining $J_{\perp,\pm} = J_{\perp,x} + iJ_{\perp,y}$ and similarly for $J_{\parallel,\pm}$ and $L_\pm$, we have
\begin{align}
    J_\pm (f \eh) &= \big(L_\pm(f) + f J_{\parallel,\pm}+ f J_{\perp,\pm}\big)\eh \\
    &= e^{\pm i\phi}\Big[\Big(\pm i \partial_\theta + i \frac{\cos \theta}{\sin \theta} \, \partial_\phi + \frac{h}{\sin \theta}\Big)f\Big ]\epm,
\end{align}
so that
\begin{equation}
    J_\pm' = e^{\pm i \phi}\Big(\pm \partial_\theta + i\frac{\cos \theta}{\sin \theta} \partial_\phi + \frac{h}{\sin \theta} \Big).
\end{equation}

To calculate $J^2$ we note that by Eq. (\ref{eq:SWSH:J_perp}), $\mbf{J}_\perp \cdot \mbf{J}_\parallel = \mbf{J}_\parallel \cdot \mbf{J}_\perp = 0$ so
\begin{equation}
    J^2 = (\mbf{J}_\parallel + \mbf{J}_\perp)^2 = J_\parallel^2 + J_\perp^2
\end{equation}
and thus
\begin{equation}
    J^2 (f\eh) = h^2 f\eh + J_\perp^2(f \eh).
\end{equation} 
Application of the product rule gives
\begin{align}\label{eq:SWSH:J2_big}
    J^2(f\eh) = &h^2f\eh + (L^2 f)\eh + |h|f (J_\perp^2 \epm) \otimes \uv{e}^{(1)}_h  \\
    &+\sum_{p=1}^3 2|h|(L_p f)(J_{\perp,p} \epm) \otimes \uv{e}^{(1)}_h + (h^2 - |h|)f(J_{\perp,p} \epm \otimes J_{\perp,p} \epm \otimes \uv{e}^{(2)}_\pm)
\end{align}
where
\begin{align}
    \uv{e}^{(q)}_h \doteq \underbrace{\epm \otimes \cdots \otimes \epm}_{|h|-q \text{ times}}.
\end{align}
From Eqs. (\ref{eq:SWSH:J_perp_x})-(\ref{eq:SWSH:J_perp_z}),
\begin{gather}
    J_\perp^2 \epm = \cot^2\theta \epm \\
    \sum_{p=1}^3 (J_{\perp,p} \epm \otimes J_{\perp,p} \epm \otimes \uv{e}^{(2)}_\pm) = \cot^2\theta \eh \\
    \sum_{p=1}^3 |h|(L_p f)(J_{\perp,p} \epm) \otimes \uv{e}^{(1)}_h = -h \frac{\cos \theta}{\sin^2 \theta}(L_z f) \eh.
\end{gather}
Substituting these into Eq. (\ref{eq:SWSH:J2_big}) gives
\begin{equation}
    J^2(f\eh) = \Big[\nabla^2 f + \frac{h^2}{\sin^2 \theta} f -2h \frac{\cos \theta}{\sin ^2 \theta} L_z f\Big]\eh.
\end{equation}
Thus we find that the helicity $h$ angular momentum operators in the $\eh$ basis are given by
\begin{align}
    J_z' &= -i \partial_\phi, \\
    J_\pm' &= e^{\pm i \phi}\Big(\pm \partial_\theta + i\frac{\cos \theta}{\sin \theta} \partial_\phi + \frac{h}{\sin \theta} \Big), \\
    J'^2 &= -\nabla^2 - \frac{2 h\cos \theta}{\sin^2 \theta}L_z + \frac{h^2}{\sin ^2 \theta},
\end{align}
so that Eq. (\ref{eq:SWSH:AM_bundle}) reduces to solving
\begin{subequations}
\begin{align}
    J_z' \, _h{f}_{jm} &= m\, _h{f}_{jm}, \\
    J_\pm' \, _h{f}_{lm} &= [(j \mp m)(j + 1 \pm m)]^{1/2} \,  _h{f}_{jm\pm 1}, \\
    J'^2 \, _h{f}_{jm} &= j(j+1) \, _h{f}_{lm}.
\end{align}
\end{subequations}
These equations are solved by the SWSHs of spin-weight $-h$ \cite{Dray1985}
\begin{equation}
    _{h}{f}_{jm}(\theta,\phi) = {_{-h}{Y}}_{jm}(\theta,\phi),
\end{equation}
so that the simultaneous eigenstates of helicity, $J^2$, $J_z$, and energy are given by
\begin{equation}
    _h \mbf{\alpha}_{jm,\mbf{k}_0} = {_{-h}{Y}_{jm}}(\theta,\phi) \delta(|\mbf{k}|-|\mbf{k}_0|)\mbf{E}_h.
\end{equation}
The SWSHs are given explicitly by \cite{Goldberg1967}
\begin{align}
    &_h{Y}_{jm}(\theta,\phi) = \Big[ \frac{(j + m)! (j-m)! (2j+1)}{4\pi (j+h)!(j-h)!}\Big]^{\frac{1}{2}} (\sin \theta / 2)^{2j} \\ 
    &\times \sum_{q}\binom{j-h}{q}\binom{j+h}{q+h-m}(-1)^{j-q-h-m}e^{im\phi}(\cot \theta / 2)^{2q+h-m}
\end{align}
where $j \geq |h|$, $-j \leq m \leq j$ and the sum ranges from $q = \max (0,m-h)$ to $q = \min (j-h,j+m)$. Interestingly, Dray \cite{Dray1985} discovered this in the inverse context, reverse engineering angular momentum operators $J'_z, J'_\pm,J'^2$ for the SWSHs. Here, we instead naturally encounter these angular momentum operators in the study of massless particles, demonstrating one of their physical origins. An arbitrary section $\alpha(\kvec)$ of $\gamma_h$ can be expressed by
\begin{equation}\label{eq:SWSH:SWSH_basis}
    \mbf{\alpha}(\kvec) = \sum_{j= |h|}^{\infty} \sum_{m=-j}^j a_{jm}(|\kvec|)\,{_{-h}Y_{jm}(\theta,\phi)}.
\end{equation}
Technically, this follows from the equivalence of sections of $\gamma_h$ with spin-weight $-h$ functions established in Sec. \ref{sec:SWSH:Spin_weighted_functions} and the completeness of the SWSHs with respect to the spin-weighted functions which was established by Newman and Penrose \cite{Newman1966}. The coefficients $a_{jm}$ can be calculated from the orthogonality of the SWSHs \cite{Goldberg1967}:
\begin{equation}
    \int_0^{2\pi} d\phi\int_{-1}^1 d(\cos(\theta)) \, {_hY}_{jm}(\theta,\phi)\, {_{h'}Y_{j'm'}} = \delta_{jj'}\delta_{mm'}.
\end{equation}

We also see how the SWSHs $_{-h}{Y}_{jm}$ are better understood as sections of the bundle $\gamma_h$ rather than as scalar functions. For example, the $_{-h}{Y}_{jm}(\theta,\phi)$ have seemingly anomalous singularities at $_{-h}{Y}_{jh}(0,\phi)$ and $_{-h}{Y}_{j,-h}(\pi,\phi)$. However, a closer examination shows that 
\begin{align}
   _{-h}{Y}_{jh}(\theta,\phi) &\sim (-1)^h \sqrt{\frac{2j+1}{4\pi}} e^{ih\phi},\;\; (\theta \rightarrow 0^+)\\
    _{-h}{Y}_{j,-h}(\theta,\phi) &\sim (-1)^{j}\sqrt{\frac{2j+1}{4\pi}} e^{-ih\phi},\;\; (\theta \rightarrow \pi^-)
\end{align}
so that by  Eq. (\ref{eq:SWSH:limiting_behavior}), the sections ${_{-h}Y_{jm}}\mbf{E}_h$ are indeed smooth at the poles; the apparent singularities of the SWSHs vanish when they are properly considered as sections rather than functions. We note that the relationship between SWSHs and sections of vector bundles was rigorously studied in the mathematics literature by Eastwood and Tod \cite{Eastwood1982}, who showed that the SWSHs can be considered as sections of line bundles over the complex projective line $\mathbb{P}^1(\Comp)$. This is essentially the same structure we encounter since $\mathbb{P}^1(\Comp)$ is homeomorphic to the Riemann sphere $S^2$, and thus the line bundles over $\mathbb{P}^1(\Comp)$ are in bijective correspondence with the $\gamma_h$. Our findings show a physical interpretation of the sectional nature of the SWSHs, showing that they are angular momentum eigenstates of massless particles. 

Eq. (\ref{eq:SWSH:SWSH_basis}) shows that the SWSHs give a countable basis for monochromatic massless waves. The vector bundle description also furnishes a basis in terms of the momentum eigenstates,  $(\khat,\uv{e}_h(\khat))$, however, this basis has certain disadvantages. First, the basis vectors are labeled by $\khat \in S^2$ and thus uncountable in number, making this a much larger basis than that given by the SWSHs. This corresponds to the spectrum of $\mbf{k}$ being continuous whereas that of $\mbf{J}^2$  and $J_z$ is discrete. Second, the momentum basis has discontinuities at the poles due to the singularities in $\uv{e}_h$; such singularities cannot be avoided by replacing $\uv{e}_h$ with some other choice of unit vectors due to the topological nontriviality of $\gamma_h$ (except when $h\neq0$). In contrast, the basis given by ${_{-h}Y_{jm}}\uv{e}_h$ is countable and composed of globally smooth states. 

In the case of photons, this results in an expansion of polarization states in terms of the angular momentum eigenstates $_{\mp1}Y_{jm}\mbf{e}_{\pm1}$. It is notable that this expansion is quite different form that used in Refs. \cite{Zaldarriaga1997,ng1999,Wiaux2006} for analyzing the CMB in terms of SWSHs. Indeed, these studies, as well as a related application in computer graphics \cite{Yi2024}, actually use the spin-weight $\pm2$ harmonics $_{\pm2}Y_{jm}$. This is due to the fact that they are analyzing Stokes vectors, which encode time averaged polarization intensities, and thus do not have any phase information. Furthermore, they are only analyzing linear polarizations since the CMB does not contain circularly polarized light.\footnote{Unpolarized light cannot become circularly polarized via Thomson scattering and is thus absent from the CMB \cite{Zaldarriaga1997}.} It turns out that linear polarizations without phase information can essentially be described by helicity $\pm2$ representations and can thus be expressed in terms of the spin-weight $\pm2$ spherical harmonics \cite{Zaldarriaga1997}. In contrast, we are analyzing the eigenstates of the full polarization space, allowing for circular polarizations and retaining phase information. In this context, photons are described by the spin-weight $\pm1$ harmonics.

\section{The nonexistence of spin and orbital angular momentum operators of massless particles}\label{sec:SWSH:SAM_OAM}
There has been a long controversy \cite{Akhiezer1965,Jaffe1990,VanEnk1994_EPL_1,VanEnk1994_JMO_2,Chen2008,Wakamatsu2010,Bliokh2010,Bialynicki-Birula2011,Leader2013,Leader2014,Leader2016,Leader2019,Yang2022, PalmerducaQin_PT, Das2024, PalmerducaQin_GT, PalmerducaQin_SAMOAM} surrounding the question of whether or not massless particles admit an SAM-OAM decomposition. We recently proved \cite{PalmerducaQin_SAMOAM}
that it is impossible to decompose the total angular momentum operator of massless particles into SAM and OAM operators, generalizing the findings of van Enk and Nienhuis who first found that a particular attempted SAM-OAM decomposition did not result in legitimate angular momentum operators \cite{VanEnk1994_EPL_1,VanEnk1994_JMO_2}. This result holds even if one significantly relaxes axiomatic constraints on the SAM and OAM operators, as shown in the no-go theorems in Ref. \cite{PalmerducaQin_SAMOAM}. However, the proofs of these theorems are somewhat abstract. Here, we can use the properties of the total angular momentum eigenstates to give a concrete illustration of the no-go results, showing that the structure of the total angular momentum multiplets given by the SWSHs rules out any reasonable spin-orbital decomposition of $\mbf{J}$.

The (nonprojective) irreducible representations of $\SO(3)$ are the complex vector spaces $V_a$ of dimension $2a+1$ where $a$ is a nonnegative integer. The tensor product of two such representations with $a\geq b$ is given by the Clebsch--Gordan addition of angular momentum:
\begin{equation}\label{eq:SWSH:AM_addition}
    V_a \otimes V_b \cong V_{a-b} \oplus V_{a-b + 1} \oplus \cdots \oplus V_{a+b}.
\end{equation}
We have that
\begin{align}
    L^2(\gamma_h) &= L^2(\mathbb{R}^+)\otimes \big(V_{|h|} \oplus V_{|h|+1} \oplus V_{|h|+2} \oplus \cdots\big) \label{eq:SWSH:massless_multiplet}\\
    &= L^2(\mathbb{R}^+) \otimes \Big(\bigoplus_{n = 0}^{\infty}V_{|h|+n}\Big)
\end{align}
where the copy of $V_{|h|+n}$ is spanned by ${_{-h}Y_{|h|+n,m}}\,\mbf{e}_h$ for $-(|h|+n)\leq m\leq|h|+n$ and $L^2(\Real^+)$ describes the radial dependence, \emph{i.e.}, the energy of the states (which is unimportant for the analysis of angular momentum). There is precisely one total angular momentum $j$ multiplet if $j\geq |h|$ and zero multiplets for $j<|h|$. This multiplet structure is distinctly different from what one typical obtains for a particle whose total angular momentum splits into well-defined spin and orbital parts as $\mbf{J} = \mbf{L} + \mbf{S}$. Consider, for example, a massive spin $s$ particle. In nonrelativistic quantum mechanics, the Hilbert space is given by
\begin{align}\label{eq:SWSH:massive_hilbert_space}
    \mathcal{H} = L^2(\Real^3) \otimes \Comp^{2s+1} = \underbrace{L^2(\Real^+)}_{\mathcal{H}_r}\otimes \underbrace{L^2(S^2)}_{\mathcal{H}_o} \otimes \underbrace{\Comp^{2s+1}}_{\mathcal{H}_s}.
\end{align}
$\mathcal{H}$, $L^2(S^2)$ and $\Comp^{2s+1}$ are each representations $\SO(3)$ whose generators are the total ($\mbf{J}$), orbital ($\mbf{L}$), and spin ($\mbf{S}$) angular momentum operators and are related by $\mbf{J} = \mbf{L} + \mbf{S}$. $\mathcal{H}_r$ is a trivial representation of $\SO(3)$ and has generator $0$. We write 
\begin{equation}
    \mathcal{H}_{J} = \mathcal{H}_o \otimes \mathcal{H}_s
\end{equation}
for the part of the Hilbert space containing nontrivial rotational symmetry. $\mathcal{H}_s = \Comp^{2s+1}$ is the irreducible representation $V_{s}$ while $L^2(S^2)$, the space of square-integrable functions on the sphere, decomposes as
\begin{equation}
    L^2(S^2) \cong V_0 \oplus V_1 \oplus V_2 \oplus \cdots
\end{equation}
where $V_l$ corresponds to the subspace spanned by the ordinary spherical harmonics $Y_{lm}$ with $|m| \leq l$. By the addition of angular momentum,  for massive particles with integer spin we have
\begin{align}
    \mathcal{H}_{J} &\cong V_s \otimes (V_0 \oplus V_1 \oplus V_2 \oplus \cdots)  
    \label{eq:SWSH:addition_of_AM_1} \\
    &\cong \Big(\bigoplus_{m=0}^{s-1}(2m+1)V_m \Big)\oplus \Big(\bigoplus_{n={s}}^{\infty}(2s+1)V_n\Big). \label{eq:SWSH:massive_multiplet}
\end{align}
The last line follows since for $0\leq m < s$, one copy of $V_{m}$ will result from $V_s \otimes V_{s+q}$ precisely when $-m\leq q\leq m$, and thus there are $2m+1$ copies of $V_m$ in $\mathcal{H}$. Meanwhile, for $n\geq s$, one copy of $V_n$  results from $V_s \otimes V_{n+q}$ precisely when $-s \leq q \leq s$, giving $2s+1$ copies of $V_n$ in $\mathcal{H}$.

Comparing Eqs. (\ref{eq:SWSH:massless_multiplet}) and (\ref{eq:SWSH:massive_multiplet}), we see that massive particles (which possess spin) and massless particles (which possess helicity) have very different angular momentum structures. For example, a massive spin $1$ particle has the multiplet structure
\begin{equation}
    \mathcal{H}_{J}^{s=1} \cong V_0 \oplus 3V_1 \oplus 3V_2 \oplus 3V_3 \oplus \dots
\end{equation}
while a massless particle with helicity $1$ has
\begin{equation}
    \mathcal{H}_{J}^{h=1} \cong V_1 \oplus V_2 \oplus V_3 \oplus  \cdots .
\end{equation}
Notably, the helicity of a massless particle enforces a lower bound on the total angular momentum which does not exist for massive particles. This is reflected in the properties of the SWSHs ${_{-h}Y_{jm}}$, which only exist if $j\geq |h|$. Massless particles also only have at most a single angular momentum multiplet of each total angular momentum $j$, while massive particles have up to $2s+1$ degenerate multiplets for each $j$. This is related to the fact that for massive particles, the spherical harmonics describe the orbital angular momentum which are then tensor multiplied by the internal spin space $\Comp^{2s+1}$, producing additional copies of the multiplets. In contrast, the SWSHs ${_{-h}Y_{jm}}$ describe the \emph{total} angular momentum states for massless particles.

From examining the addition of angular momentum formula in Eq. (\ref{eq:SWSH:AM_addition}), it becomes apparent that the multiplet structure of massless particles [Eq. (\ref{eq:SWSH:massless_multiplet})] would not result from any reasonable spin-orbital splitting of the angular momentum. In particular, it is difficult to avoid generating repeated multiplets and or low angular momentum multiplets. Thus, from examining the difference between massive and massless angular momentum states, and the differences between ordinary and spin-weighted spherical harmonics, we come to the conclusion that massless particles should not admit an SAM-OAM decomposition.

We note that it is technically possible to contrive spin and orbital multiplet structures which would tensor product to give Eq. (\ref{eq:SWSH:AM_addition}), but they are highly unnatural. For example, the simplest such splitting for a helicity $1$ massless particle is
\begin{equation}\label{eq:SWSH:unnaatural splitting}
    V_1 \otimes(V_0\oplus V_3 \oplus V_6 \oplus \cdots) \cong V_1 \oplus V_2 \oplus V_3 \oplus \cdots
\end{equation}
in which the candidate orbital angular momentum $l$ could only take on integer multiples of $3$, a strange hypothetical which has never been suggested. We note that one can rigorously prove that the rotationally symmetry of the massless Hilbert spaces $L^2(\gamma_{h,S^2})$ cannot be described by a tensor product of two nontrivial $\SO(3)$ representations (\cite{PalmerducaQin_SAMOAM}, No-Go Theorem 1), so these exotic splittings are not realized in massless particles.

We note that the massless SAM-OAM obstruction can also be understood by examining the complete set of commuting observables (CSCO). As we saw in the previous section, $(H,\chi,J^2, J_z)$ form a CSCO for massless particles, that is, a basis of states can be labeled by $|\kvec_0|,h,j,j_z$. Compare this with massive particles, for which we know the angular momentum splits into well-defined SAM and OAM parts:
\begin{equation}\label{eq:SWSH:massive_splitting}
    \mbf{J} = \mbf{S} + \mbf{L}.
\end{equation}
Since $\mbf{S}$ and $\mbf{L}$ are $\SO(3)$ generators, it follows from the representation theory of $\SO(3)$ that $S^2$ and $L^2$ commute with each other and with $J^2$, and $J_z$. It is also true that $\mbf{S}$ and $\mbf{L}$ commute with $H$, and therefore so do $S^2$ and $L^2$ \cite{Terno2003}. Indeed, a CSCO for massive particles with spin, assuming no additional internal structure (e.g., no color charge), is given by $(H, J^2, J_z,S^2,L^2)$, and a basis is labeled by $(|\kvec_0|, j,j_z,s,l)$ \cite{Shankar_QM}. $s$ is determined by the identity of the particles, as is $h$ in the massless case. We see though a discrepancy between massive and massless particles: one must specify both the orbital and total angular momentum eigenvalues $l$ and $j$ in the massive case, whereas in the massless case only $j$ is needed. 

Suppose for a massless particle we assume that Eq. (\ref{eq:SWSH:massive_splitting}) is true for some angular momentum operators $\mbf{S}$ and $\mbf{L}$ which commute with $H$. It then it follows that $L^2$ commutes with $(H,\chi,J^2,J_z)$. However, the latter is already a complete set of observables, so it must be that $L^2$ gives redundant information such that the value of $l$ can be inferred from $|\kvec_0|,h,j,j_z$. This corresponds precisely to the unnatural splittings of the form Eq. (\ref{eq:SWSH:unnaatural splitting}), since each $V_j$ on the rhs (labeled by j) is produced by a single tensor product $V_1 \otimes V_l$ on the lhs. As previously mentioned, these unnatural splittings can be rigorously ruled out by the no-go theorems of Ref. \cite{PalmerducaQin_SAMOAM}.

\section{Conclusion}
Massive and massless particles have very different Hilbert spaces. Massive particles are topologically trivial, allowing the internal and external DOFs to be separated. The external Hilbert space decomposes into rotationally invariant subspaces spanned by the ordinary spherical harmonics. Massless particles, on the other hand, are topologically nontrivial (for $h \neq 0$). Even though they only have a single internal DOF, this DOF is twisted together with the external momentum DOFs. The result is that states are described not by ordinary functions, but by sections of topologically nontrivial line bundles (or equivalently by spin-weighted functions). The rotationally invariant subspaces are spanned by the total angular momentum multiplets, which in this case are SWSHs $({_{-h}Y_{jm}})$ rather than ordinary spherical harmonics $(Y_{jm})$. This harmonic decomposition gives a smooth countable basis for monochromatic massless waves, something which is not provided by the momentum eigenstates. The angular momentum multiplet structure of the SWSHs is also quite different than that which occurs for massive particles. Indeed, in the massless case we obtain a very sparse multiplet structure, with a single multiplet for each $j\geq |h|$ and no multiplet if $j < |h|$. In contrast, there is no positive lower bound on the angular momentum of massive particles with integer spin and typical multiplets are degenerate (except in the case of spin $0$ particles). From taking into account the way angular momenta add, one sees that the massless multiplet structure cannot result from any reasonable SAM-OAM decomposition of the angular momentum $\mbf{J}$. Indeed, it was recently established via more abstract geometric and topological arguments that such a massless SAM-OAM decomposition is impossible \cite{PalmerducaQin_SAMOAM}; here we see that this conclusion is supported by carrying out the explicit angular momentum decomposition of massless particles.

\acknowledgments

This work is supported by U.S. Department of Energy (DE-AC02-09CH11466).


\bibliographystyle{JHEP}
\bibliography{biblio.bib}


\end{document}